\documentclass{llncs}

\usepackage{amssymb}
\usepackage{amsmath}
\usepackage{listings}
\usepackage{microtype}
\usepackage{wrapfig}
\usepackage{float}

\usepackage{tikz}
\usepackage{tikzscale}
\usepackage{pgfplots}
\pgfplotsset{compat=1.15}
\usepgfplotslibrary{groupplots}
\usetikzlibrary{matrix}
\usetikzlibrary{fit}
\usetikzlibrary{positioning}

\makeatletter
\pgfplotsset{
    groupplot xlabel/.initial={},
    every groupplot x label/.style={
        at={($({\pgfplots@group@name\space c1r\pgfplots@group@rows.west}|-{\pgfplots@group@name\space c1r\pgfplots@group@rows.outer south})!0.5!({\pgfplots@group@name\space c\pgfplots@group@columns r1.east}|-{\pgfplots@group@name\space c1r\pgfplots@group@rows.outer south})$)},
        anchor=north,
    },
    groupplot ylabel/.initial={},
    every groupplot y label/.style={
            rotate=90,
        at={($({\pgfplots@group@name\space c1r1.north}-|{\pgfplots@group@name\space c1r1.outer
west})!0.5!({\pgfplots@group@name\space c1r\pgfplots@group@rows.south}-|{\pgfplots@group@name\space c1r\pgfplots@group@rows.outer west})$)},
        anchor=south
    },
    execute at end groupplot/.code={%
      \node [/pgfplots/every groupplot x label]
{\pgfkeysvalueof{/pgfplots/groupplot xlabel}};  
      \node [/pgfplots/every groupplot y label] 
{\pgfkeysvalueof{/pgfplots/groupplot ylabel}};  
    }
}

\def\endpgfplots@environment@groupplot{%
    \endpgfplots@environment@opt%
    \pgfkeys{/pgfplots/execute at end groupplot}%
    \endgroup%
}
\makeatother

\title{Pattern-defeating Quicksort}
\author{Orson R. L. Peters\\orsonpeters@gmail.com}
\institute{Leiden University}

\begin{document}

\maketitle

\begin{abstract}

    A new solution for the Dutch national flag problem is proposed, requiring no three-way
    comparisons, which gives quicksort a proper worst-case runtime of $O(nk)$ for inputs with $k$
    distinct elements. This is used together with other known and novel techniques to construct a
    hybrid sort that is never significantly slower than regular quicksort while speeding up
    drastically for many input distributions.

\end{abstract}

\section{Introduction}

Arguably the most used hybrid sorting algorithm at the time of writing is
introsort\cite{Musser97introspectivesorting}. A combination of insertion sort,
heapsort\cite{williams1964algorithm} and quicksort\cite{hoare1962quicksort}, it is very fast and can
be seen as a truly hybrid algorithm. The algorithm performs introspection and decides when to change
strategy using some very simple heuristics. If the recursion depth becomes too deep, it switches to
heapsort, and if the partition size becomes too small it switches to insertion sort.

The goal of pattern-defeating quicksort (or \textit{pdqsort}) is to improve on introsort's
heuristics to create a hybrid sorting algorithm with several desirable properties. It maintains
quicksort's logarithmic memory usage and fast real-world average case, effectively recognizes and
combats worst case behavior (deterministically), and runs in linear time for a few common patterns.
It also unavoidably inherits in-place quicksort's instability, so pdqsort can not be used in
situations where stability is needed.

In Section~\ref{sec:overview} we will explore a quick overview of pattern-defeating quicksort and related work,
in Section~\ref{sec:flag} we propose our new solution for the Dutch national flag problem and prove its $O(nk)$
worst-case time for inputs with $k$ distinct elements, Section~\ref{sec:novel} describes other novel
techniques used in pdqsort while Section~\ref{sec:known} describes various previously known ones. Our final
section consists of an empirical performance evaluation of pdqsort.

This paper comes with an open source state of the art C++ implementation\cite{pdqsortgithub}. The
implementation is fully compatible with \texttt{std::sort} and is released under a permissive
license. Standard library writers are invited to evaluate and adopt the implementation as their
generic unstable sorting algorithm. At the time of writing the Rust programming language has adopted
pdqsort for \texttt{sort\_unstable} in their standard library thanks to a porting effort by Stjepan
Glavina. The implementation is also available in the C++ \texttt{Boost.Sort} library.

\section{Overview and related work}
\label{sec:overview}

Pattern-defeating quicksort is a hybrid sort consisting of quicksort\cite{hoare1962quicksort},
insertion sort and a fallback sort. In this paper we use heapsort as our fallback sort, but really
any $O(n \log n)$ worst case sort can be used - it's exceptionally rare that the heuristics end up
switching to the fallback. Each recursive call of pdqsort chooses either to fall back, use
insertion sort or partition and recurse.

Insertion sort is used for small recursive calls as despite its $O(n^2)$ worst case it has great
constant factors and outperforms quicksort for small $n$. We later discuss (not novel but
nevertheless important) techniques to properly implement insertion sort for usage in a hybrid
algorithm. We have tried to use small sorting networks akin to Codish's\cite{Codish2017} approach as
an alternative base case but were unable to beat insertion sort. We conjecture that the small code
size of insertion sort has sufficient positive cache effects to offset the slower algorithm when
used in a hybrid sort.

For partitioning we use a novel scheme that indirectly ends up performing tripartite partitioning.
This is used in conjunction with the very important technique from BlockQuicksort\cite{EdelkampW16}
that greatly speeds up partitioning with branchless comparison functions. In a sense our
partitioning scheme is similar to Yaroslavskiy's dual-pivot quicksort\cite{dualpivot} from the
perspective of equal elements. We did consider dual- and multi-pivot\cite{multipivot} variants of
quicksort but chose to stick to traditional partitioning for simplicity, applicability of the
techniques described here, and due to the massive speedup from BlockQuicksort, which does not
trivially extend to multiple pivots (see \textsc{IPS$^4$o}\cite{ips4o} for that).

We use the well-known median-of-3 pivot selection scheme, with John Tukey's
ninther\cite{tukey1978ninther} for sufficiently large inputs, which Kurosawa\cite{Kurosawa16} finds
has a near-identical number of comparisons to selecting the true median, but is significantly
simpler. 

Finally, we beat patterns with two novel additions. We diverge from introsort by no longer simply
using the call depth to switch to a fallback. Instead we define a concept of a bad partition, and
track those instead. This results in a more precise heuristic on bad sorting behavior, and
consequently fallback usage.  Whenever we detect a bad partition we also swap a couple well-selected
elements, which not only breaks up common patterns and introduces 'noise' similar to how a random
quicksort would behave, it introduces new pivot candidates in the selection pool. We also use a
technique (not known to us in previous literature) due to Howard Hinnant to optimistically handle
ascending/descending sequences with very little overhead.

\section{A faster solution for the Dutch national flag problem}
\label{sec:flag}

A naive quicksort implementation might trigger the $\Theta(n^2)$ worst case on the all-equal input
distribution by placing equal comparing elements in the same partition. A smarter implementation
either always or never swaps equal elements, resulting in average case performance as equal elements
will be distributed evenly across the partitions. However, an input with many equal comparing
elements is rather common\footnote{It is a common technique to define a custom comparison function
that only uses a subset of the available data to sort on, e.g. sorting cars by their color. Then you
have many elements that aren't fundamentally equal, but do compare equal in the context of a sorting
operation.}, and we can do better. Handling equal elements efficiently requires tripartite
partitioning, which is equivalent to Dijkstra's Dutch national flag
problem\cite{Dijkstra:1997:DP:550359}.

\begin{wrapfigure}{r}{0.4\textwidth}
    \centering
    \includegraphics[width=0.38\textwidth]{graphics/bentley_invariant.tikz}
    \caption{The invariant used by Bentley-McIlroy. After partitioning the equal elements stored at the
    beginning and at the end are swapped to the middle.}

    \includegraphics[width=0.38\textwidth]{graphics/pdqsort_invariant.tikz}

    \caption{The invariant used by \texttt{partition\_right} of pdqsort, shown at respectively the
    initial, halfway and finished state. When the loop is done the pivot gets swapped into its
    correct position. $p$ is the single pivot element. $r$ is the pointer returned by
    the partition routine indicating the pivot position. The dotted lines indicate how $i$ and $j$
    change as the algorithm progresses. This is a simplified representation, e.g. $i$ is actually off by
    one.}
\end{wrapfigure}

Pattern-defeating quicksort uses the fast 'approaching pointers' method\cite{bentley1993engineering}
for partitioning. Two indices are initialized, $i$ at the start and $j$ at the end of the sequence.
$i$ is incremented and $j$ is decremented while maintaining an invariant, and when both invariants
are invalidated the elements at the pointers are swapped, restoring the invariant.  The algorithm
ends when the pointers cross. Implementers must take great care, as this algorithm is conceptually
simple, but is very easy to get wrong.

Bentley and McIlroy describe an invariant for partitioning that swaps
equal elements to the edges of the partition, and swaps them back into the middle after
partitioning. This is efficient when there are many equal elements, but has a significant drawback.
Every element needs to be explicitly checked for equality to the pivot before swapping, costing
another comparison. This happens regardless of whether there are many equal elements, costing
performance in the average case.

Unlike previous algorithms, pdqsort's partitioning scheme is not self contained. It uses two
separate partition functions, one that groups elements equal to the pivot in the left partition
(\texttt{partition\_left}), and one that groups elements equal to the pivot in the right partition
(\texttt{partition\_right}). Note that both partition functions can always be implemented using a
single comparison per element as $a < b \Leftrightarrow a \ngeq b$ and $a \nless b \Leftrightarrow a
\geq b$.

For brevity we will be using a simplified, incomplete C++ implementation to illustrate pdqsort.  It
only supports \texttt{int} and compares using comparison operators. It is however trivial to extend
this to arbitrary types and custom comparator functions. To pass subsequences\footnote{Without
exception, in this paper subsequences are assumed to be contiguous.} around, the C++ convention is
used of one pointer at the start, and one pointer at one-past-the-end. For the exact details refer
to the full implementation\cite{pdqsortgithub}.

Both partition functions assume the pivot is the first element, and that it has been selected as a
median of at least three elements in the subsequence. This saves a bound check in the first
iteration.

\begin{figure}[t]
\centering
\begin{lstlisting}[language=C++,basicstyle=\ttfamily\footnotesize,columns=flexible]
int* part_left(int* l, int* r) {           int* part_right(int* l, int* r) {  
    int* i = l; int* j = r;                    int* i = l; int* j = r; 
    int p = *l;                                int p = *l;
                                                                                        
    while (*--j > p);                          while (*++i < p);                   
    if (j + 1 == r) {                          if (i - 1 == l) {                   
        while (i < j && *++i <= p);                while (i < j && *--j >= p);
    } else {                                   } else {                            
        while (*++i <= p);                         while (*--j >= p);              
    }                                          }                                   
                                                                                   
                                               // bool no_swaps = (i >= j);        
    while (i < j) {                            while (i < j) {                     
        std::swap(*i, *j);                         std::swap(*i, *j);              
        while (*--j > p);                          while (*++i < p);               
        while (*++i <= p);                         while (*--j >= p);              
    }                                          }                                   
                                                                                   
    std::swap(*l, *j);                         std::swap(*l, *(i - 1));            
    return j;                                  return i - 1;                       
}                                          }                                       
\end{lstlisting}
\caption{An efficient implementation of \texttt{partition\_left} and \texttt{partition\_right}
(named \texttt{part\_left} and \texttt{part\_right} here due to limited page width). Note the
(almost) lack of bound checks, we assume that \texttt{p} was selected as the median of at least
three elements, and in later iterations previous elements are used as sentinels to prevent going out
of bounds. Also note that a pre-partitioned subsequence will perform no swaps and that it is
possible to detect this with a single comparison of pointers, \texttt{no\_swaps}. This is used for a
heuristic later.}
\end{figure}

Given a subsequence $\alpha$ let us partition it using \texttt{partition\_right} using pivot $p$.
We then inspect the right partition, calling its first element $q$ and the remaining subsequence of
elements $\beta$:

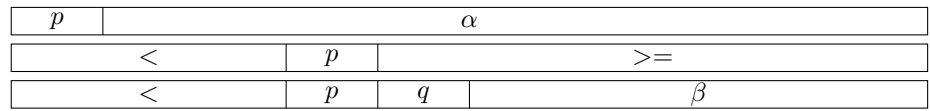
\vspace*{.5em}
\noindent \includegraphics[width=\textwidth]{graphics/alt_expl.tikz}

If $p \neq q$ we have $q > p$, and apply \texttt{partition\_right} on $q, \beta$. Rename
$q, \beta$ to be the left partition of this operation (marked as $q', \beta'$ in the diagram to
emphasize renaming). The right partition is marked as '$>$', because in this process we have the
perspective of pivot $p$, but it's definitely possible for elements equal to $q$ to be in the
partition marked '$>$'.

\vspace*{.5em}
\noindent \includegraphics[width=\textwidth]{graphics/alt_expl1.tikz}

We apply the above step recursively as long as $p \neq q$. If at some point $q, \beta$ becomes
empty, we can conclude there were no elements equal to $p$ and the tripartite partitioning was done
when we initially partitioned $\alpha$. Otherwise, consider $p = q$. We know that $\forall x \in
\beta\colon x \geq p$, thus $\nexists x \in \beta\colon x < q$. If we were to partition $q,\beta$
using \texttt{partition\_left}, any element smaller than or equal to $q$ would be partitioned left.
However, we just concluded that $\beta$ can not contain elements smaller than $q$.  Thus, $\beta$'s
left partition only contains elements equal to $q$ (and thus equal to $p$), and its right partition
only contains elements bigger than $q$:

\vspace*{.5em}
\noindent \includegraphics[width=\textwidth]{graphics/alt_expl2.tikz}

\qed

This leads to the partitioning algorithm used by pdqsort. The \textit{predecessor} of a
subsequence is the element directly preceding it in the original sequence. A subsequence that is
\textit{leftmost} has no predecessor. If a subsequence has a predecessor $p$ that compares equal to
the chosen pivot $q$, apply \texttt{partition\_left}, otherwise apply \texttt{partition\_right}. No
recursion on the left partition of \texttt{partition\_left} is needed, as it contains only
equivalent elements.

\subsection{An $O(nk)$ worst case of pdqsort with $k$ distinct elements}

\begin{lemma}
Any predecessor of a subsequence\footnote{Assuming only those subsequences that are passed into
recursive calls of pdqsort.} was the pivot of an ancestor.
\end{lemma}
\begin{proof}
If the subsequence is a direct right child of its parent partition, its predecessor is the pivot of
the parent. However, if the subsequence is the left child of its parent partition, its predecessor
is the predecessor of its parent. Since our lemma states that our subsequence has a predecessor, it
is not leftmost and there must exist some ancestor of which the subsequence is a right child.
\end{proof}

\begin{lemma}
The first time a distinct value $v$ is selected as a pivot, it can't be equal to its predecessor.
\end{lemma}
\begin{proof}
Assume $v$ is equal to its predecessor. By Lemma 1 this predecessor was the pivot of an ancestor
partition. This is a contradiction, thus $v$ is not equal to its predecessor.
\end{proof}
\begin{corollary}
The first time $v$ is selected as a pivot, it is always used to partition with
{\normalfont\texttt{partition\_right}}, and all elements $x$ such that $x = v$ end up in the right partition.
\end{corollary}

\begin{lemma}
Until an element equal to $v$ is selected as a pivot again, for all $x$ comparing equal to $v$, $x$
must be in the partition directly to the right of $v$.
\end{lemma}
\begin{proof}
By Corollary 1, $x$ can not be in a partition to the left of $v$ and thus any other partitions with
pivot $w < v$ are irrelevant. There also can't be a partition with pivot $w > v$ that puts any $x$
equal to $v$ in its right partition, as this would imply $x \geq w > v$, a contradiction. So any
elements $x = v$ stay in a partition directly to the right of $v$.
\end{proof}

\begin{lemma}
The second time a value equal to $v$ is selected as a pivot, all $x = v$ are in the correct position
and are not recursed upon any further.
\end{lemma}
\begin{proof}
The second time another element $x = v$ is selected as a pivot, Lemma 3 shows that $v$ must be
its predecessor, and thus it is used to partition with \texttt{partition\_left}. In this
partitioning step all elements equal to $x$ (and thus equal to $v$) end up in the left partition,
and are not further recursed upon. Lemma 3 also provides the conclusion we have just finished
processing \textbf{all} elements equal to $v$ passed to pdqsort.
\end{proof}

\begin{theorem}
pdqsort has complexity $O(nk)$ when the input distribution has $k$ distinct values.\footnote{Note
that this is an upper bound. When $k$ is big $O(n \log n)$ still applies.}
\end{theorem}
\begin{proof}
Lemma 4 proves that every distinct value can be selected as a pivot at most twice, after which every
element equal to that value is sorted correctly. Each partition operation has complexity $O(n)$. The
worst-case runtime is thus $O(nk)$.
\end{proof}
\qed

\section{Other novel techniques}
\label{sec:novel}

\subsection{Preventing quicksort's $O(n^2)$ worst case}

Pattern-defeating quicksort calls any partition operation which is more unbalanced than $p$ (where
$p$ is the percentile of the pivot, e.g. $\frac{1}{2}$ for a perfect partition) a \textit{bad
partition}.  Initially, it sets a counter to $\log n$. Every time it encounters a bad partition, it
decrements the counter before recursing\footnote{This counter is maintained separately in every
subtree of the call graph - it is not a global to the sort process. Thus, if after the first
partition the left partition degenerates in the worst case it does not imply the right partition
also does.}. If at the start of a recursive call the counter is 0 it uses heapsort to sort this
subsequence, rather than quicksort.

\begin{lemma}
At most $O(n \log n)$ time is spent in pdqsort on bad partitions.
\end{lemma}

\begin{proof}
Due to the counter ticking down, after $\log n$ levels that contain a bad partition
the call tree terminates in heapsort. At each level we may do at most $O(n)$ work, giving a runtime
of $O(n \log n)$.
\end{proof}

\begin{lemma}
At most $O(n \log n)$ time is spent in pdqsort on good partitions.
\end{lemma}

\begin{proof}
Consider a scenario where quicksort's partition operation always puts $pn$ elements in the left
partition, and $(1-p)n$ in the right. This consistently forms the worst possible good partition.
Its runtime can be described with the following recurrence relation:
$$T(n, p) = n + T(pn, p) + T((1-p)n, p)$$
For any $p \in (0, 1)$ the Akra-Bazzi\cite{Akra1998} theorem shows
$\Theta(T(n, p)) = \Theta(n \log n)$.
\end{proof}

\begin{theorem}
Pattern-defeating quicksort has complexity $O(n \log n)$.
\end{theorem}

\begin{proof}
Pattern-defeating quicksort spends $O(n \log n)$ time on good partitions, bad partitions, and
degenerate cases (due to heapsort also being $O(n \log n)$). These three cases exhaustively
enumerate any recursive call to pdqsort, thus pattern-defeating quicksort has complexity
$O(n \log n)$.
\end{proof}
\qed

We have proven that for any choice of $p \in (0, 1)$ the complexity of pattern-defeating quicksort
is $O(n \log n)$. However, this does not tell use what a good choice for $p$ is.

Yuval Filmus\cite{yuval} solves above recurrence, allowing us to study the slowdown of quicksort
compared to the optimal case of $p = \frac{1}{2}$. He finds that the solution is
$$\lim_{n \to \infty} \frac{T(n, p)}{T(n, \frac{1}{2})} = \frac{1}{H(p)}$$
where $H$ is Shannon's binary entropy function: $$H(p) = -p \log_2(p) - (1-p)\log_2(1-p)$$

Plotting this function gives us a look at quicksort's fundamental performance characteristics:

\begin{figure}[H]
    \centering
    \begin{tikzpicture}
      \begin{axis}[ 
        domain=0:0.5,
        width=0.9\textwidth,
        height=6cm,
        restrict y to domain=1:5,
        samples=300,
        xlabel=$p$,
        ylabel=factor,
        x tick label style={/pgf/number format/fixed},
        minor xtick={0, 0.05, ..., 1},
        ytick={1, 1.5, ..., 5},
        smooth, grid style={dashed}, grid
      ] 
      \addplot [mark=none] {1/(-x*log2(x) - (1 - x)*log2(1-x))}; 
      \end{axis}
    \end{tikzpicture}

    \caption{Slowdown of $T(n, p)$ compared to $T(n, \frac{1}{2})$. This beautifully shows why
    quicksort is generally so fast. Even if every partition is split 80/20, we're still running only
    40\% slower than the ideal case.}
    \label{fig:slowdown}
\end{figure}
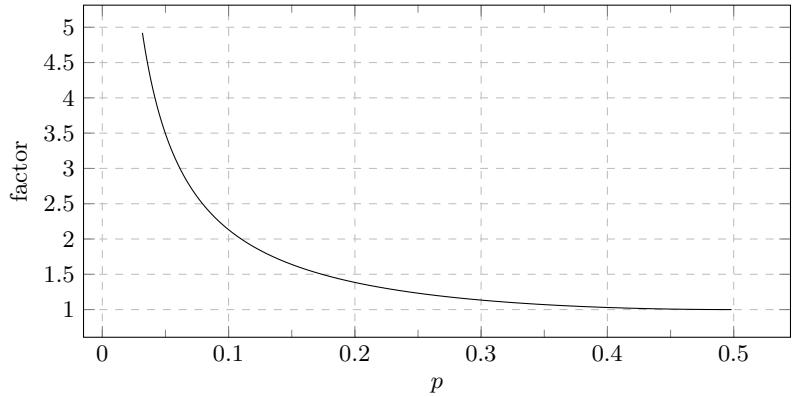

From benchmarks we've found that heapsort is roughly about twice as slow as quicksort for sorting
randomly shuffled data. If we then choose $p$ such that $H(p)^{-1} = 2$ a bad
partition becomes roughly synonymous with `worse than heapsort'.

The advantage of this scheme is that $p$ can be tweaked if the architecture changes, or you have
a different worst-case sorting algorithm instead of heapsort.

We have chosen $p = 0.125$ as the cutoff value for bad partitions for two reasons: it's
reasonably close to being twice as slow as the average sorting operation and it can be computed
using a simple bitshift on any platform.

Using this scheme as opposed to introsort's static logarithmic recursive call limit for preventing
the worst case is more precise. While testing we noticed that introsort (and to a lesser extent
pdqsort) often have a rough start while sorting an input with a bad pattern, but after some
partitions the pattern is broken up. Our scheme then procedes to use the now fast quicksort for the
rest of the sorting whereas introsort too heavily weighs the bad start and degenerates to heapsort.

\subsection{Introducing fresh pivot candidates on bad partitions}

Some input patterns form some self-similar structure after partitioning. This can cause a similar
pivot to be repeatedly chosen. We want to eliminate this.  The reason for this can be found in
Figure~\ref{fig:slowdown} as well. The difference between a good and mediocre pivot is small, so
repeatedly choosing a good pivot has a relatively small payoff. The difference between a mediocre
and bad pivot is massive. An extreme example is the traditional $O(n^2)$ worst case: repeatedly
partitioning without real progress.

The classical way to deal with this is by randomizing pivot selection (also known as randomized
quicksort). However, this has multiple disadvantages. Sorting is not deterministic, the access
patterns are unpredictable and extra runtime is required to generate random numbers. We also destroy
beneficial patterns, e.g. the technique in section 5.2 would no longer work for descending patterns
and performance on 'mostly sorted' input patterns would also degrade.

Pattern-defeating quicksort takes a different approach. After partitioning we check if the partition
was \textit{bad}. If it was, we swap our pivot candidates for others. In our implementation pdqsort
chooses the median of the first, middle and last element in a subsequence as the pivot, and swaps the
first and last candidate for ones found at the 25\% and 75\% percentile after encountering a bad
partition. When our partition is big enough that we would be using Tukey's ninther for pivot
selection we also swap the ninther candidates for ones at roughly the 25\% and 75\% percentile of
the partition.

With this scheme pattern-defeating quicksort is still fully deterministic, and with minimal overhead
breaks up many of the patterns that regular quicksort struggles with. If the downsides of
non-determinism do not scare you and you like the guarantees that randomized quicksort provides
(e.g. protection against DoS attacks) you can also swap out the pivot candidates with random
candidates. It's still a good idea to \textit{only} do this after a bad partition to prevent
breaking up beneficial patterns.

\section{Previously known techniques}
\label{sec:known}

\subsection{Insertion sort}

Essentially all optimized quicksort implementations switch to a different algorithm when recursing
on a small ($\leq 16$--$32$ elements)  subsequence, most often insertion sort. But even the
simple insertion sort is subject to optimization, as a significant amount of time is spent in this
phase.

It's important to use a series of moves instead of swaps, as this eliminates a lot of unnecessary
shuffling. Another big\footnote{5-15\% from our benchmarks, for sorting small integers.} improvement
can be made by eliminating bounds checking in the inner insertion sort loop. This can be done for
any subsequence that is not leftmost, as there must exist some element before the subsequence you're
sorting that acts as a sentinel breaking the loop. Although this change is tiny, for optimal
performance this requires an entirely new function as switching behavior based on a condition
defeats the purpose of this micro-optimization. This variant is
traditionally\cite{musser1994algorithm} called \texttt{unguarded\_insertion\_sort}. This is a prime
example where introducing more element comparisons can still result in faster code.

\subsection{Optimistic insertion sorting on a swapless partition}

This technique is not novel, and is due to Howard Hinnants\cite{libcxx} \texttt{std::sort}
implementation, but has to our knowledge not been described in literature before. After partitioning
we can check whether the partition was \textit{swapless}. This means we didn't have to swap any
element (other than putting the pivot in its proper place). Checking for this condition can be done
with a single comparison as shown in Figure 3 with the variable \texttt{no\_swaps}.

If this is the case and the partition wasn't \textit{bad} we do a partial insertion sort over both
partitions that aborts if it has to do more than a tiny number of corrections (to minimize overhead
in the case of a misdiagnosed best case). However, if little to no corrections were necessary, we
are instantly done and do not have to recurse.

If you properly choose your pivots\footnote{\label{fragile}Frankly, this can be a bit fragile. Check
the pdqsort source code for the exact procedure used to select pivots, and you notice we actually
sort the pivot candidates and put the median one at the start position, which gets swapped back to
the middle in the event of a swapless partition. A small deviation from the reference implementation
here might lose you the linear time guarantee.} then this small optimistic heuristic will sort
inputs that are ascending, descending or ascending with an arbitrary element appended all in linear
time. We argue that these input distributions are vastly overrepresented in the domain of inputs to
sorting functions, and are well worth the miniscule overhead caused by false positives.

This overhead really is miniscule, as swapless partitions become exceedingly unlikely for large
arrays to happen by chance. A hostile attacker crafting worst-case inputs is also no better off. The
maximum overhead per partition is $\approx n$ operations, so it doubles performance at worst, but
this can already be done to pdqsort by crafting a worst case that degenerates to heapsort.
Additionally, the partial insertion sort is only triggered when the partition wasn't \textit{bad},
forcing an attacker to generate good quicksort progress if she intends to trigger repeated
misdiagnosed partial insertion sorts.

\subsection{Block partitioning}

One of the most important optimizations for a modern quicksort is Edelkamp and Wei\ss' recent work
on BlockQuicksort\cite{EdelkampW16}. Their technique gives a huge\footnote{50-80\% from our
benchmarks, for sorting small integers.} speedup by eliminating branch predictions during
partitioning. In pdqsort it is only applied for \texttt{partition\_right}, as the code size is
significant and \texttt{partition\_left} is rarely called.

Branch predictions are eliminated by replacing them with data-dependent moves. First some static
block size is determined\footnote{In our implementation we settled on a static 64 elements, but the
optimal number depends on your CPU and cache architecture as well as the data you're sorting.}.
Then, until there are fewer than \texttt{2*bsize} elements remaining, we repeat the following
process.

We look at the first \texttt{bsize} elements on the left hand side. If an element in this
block is bigger or equal to the pivot, it belongs on the right hand side. If not, it should keep its
current position. For each element that needs to be moved we store its offset in \texttt{offsets\_l}.
We do the same for \texttt{offsets\_r}, but now for the last \texttt{bsize} elements, and
finding elements that are strictly less than the pivot:
\begin{lstlisting}[language=C++,basicstyle=\ttfamily\footnotesize,breaklines=true,keywordstyle=\bfseries,columns=flexible,escapechar=@]
int num_l = 0;                             int num_r = 0;                         
for (int i = 0; i < bsize; ++i) {          for (int i = 0; i < bsize; ++i) { 
    if (*(l + i) >= pivot) {                   if (*(r - 1 - i) < pivot) {        
        offsets_l[num_l] = i;                      offsets_r[num_r] = i + 1;      
        num_l++;                                   num_r++;                       
    }                                          }                                  
}                                          }                                      
\end{lstlisting}
\noindent But this still contains branches. So instead we do the following:
\begin{lstlisting}[language=C++,basicstyle=\ttfamily\footnotesize,breaklines=true,keywordstyle=\bfseries,columns=flexible,escapechar=@]
int num_l = 0;                             int num_r = 0;                    
for (int i = 0; i < bsize; ++i) {          for (int i = 0; i < bsize; ++i) { 
    offsets_l[num_l] = i;                      offsets_r[num_r] = i + 1;     
    num_l += *(l + i) >= pivot;                num_r += *(r - 1 - i) < pivot; 
}                                          }                                 
\end{lstlisting}
\noindent This contains no branches. Now we can \textbf{unconditionally} swap elements from the offset buffers:
\begin{lstlisting}[language=C++,basicstyle=\ttfamily\footnotesize,breaklines=true,keywordstyle=\bfseries,columns=flexible,escapechar=@]
for (int i = 0; i < std::min(num_l, num_r); ++i) {
    std::iter_swap(l + offsets_l[i], r - offsets_r[i]);
}
\end{lstlisting}
Notice that we only swap \texttt{std::min(num\_l, num\_r)} elements, because we need to pair
each element that belongs on the left with an element that belongs on the right. Any leftover
elements are re-used in the next iteration\footnote{After each iteration at least one offsets buffer
is empty. We fill any buffer that is empty.}, however it takes a bit of extra code to do so. It is
also possible to re-use the last remaining buffer for the final elements to prevent any wasted
comparisons, again at the cost of a bit of extra code. For the full implementation\footnote{We skip
over many important details and optimizations here as they are more relevant to BlockQuicksort than
to pattern-defeating quicksort. The full implementation has loop unrolling, swaps elements using
only two moves per element rather than three and uses all comparison information gained while
filling blocks.} and more explanation we invite the reader to check the Github repository, and read
Edelkamp and Wei\ss' original paper.

The concept is important here: replacing branches with data-dependent moves followed by
unconditional swaps. This eliminates virtually all branches in the sorting code, as long as
the comparison function used is branchless. This means in practice that the speedup is limited to
integers, floats, small tuples of those or similar. However, it's still a comparison sort.
You can give it arbitrarily complicated branchless comparison functions (e.g. \texttt{a*c > b-c})
and it will work.

When the comparison function isn't branchless this method of partitioning can be slower. The C++
implementation is conservative, and by default only uses block based partitioning if the comparison
function is \texttt{std::less} or similar, and the elements being sorted are native numeric types.
If a user wishes to get block based partitioning otherwise it needs to be specifically requested.

\section{Experimental results}

\subsection{Methodology}

We present a performance evaluation of pattern-defeating quicksort with (\textsc{bpdq}) and without
(\textsc{pdq}) block partitioning, introsort from libstdc++'s implementation of \texttt{std::sort}
(\textsc{std}), Timothy van Slyke's C++ Timsort\cite{timsort} implementation\cite{tvanslyke}
(\textsc{tim}), BlockQuicksort (\textsc{bq}) and the sequential version of In-Place Super Scalar
Samplesort\cite{ips4o} (\textsc{is$^4$o}). The latter algorithm represents to our knowledge the
state of the art in sequential in-place comparison sorting for large amounts of data.

In particular the comparison with BlockQuicksort is important as it is a benchmark for the novel
methods introduced here. The code repository for BlockQuicksort defines many different versions of
BlockQuicksort, one of which also uses Hoare-style crossing pointers partitioning and Tukey's
ninther pivot selection. This version is chosen for the closest comparison as it most resembles our
algorithm. The authors of BlockQuicksort also proposed their own duplicate handling scheme. To
compare the efficacy of their and our approach we also chose the version of BlockQuicksort with it
enabled.

We evaluate the algorithms for three different data types. The simplest is \textsc{int}, which is a
simple 64-bit integer. However, not all data types have a branchless comparison function. For that
reason we also have \textsc{str}, which is a \texttt{std::string} representation of \textsc{int}
(padded with zeroes such that lexicographic order matches the numeric order). Finally to simulate
an input with an expensive comparison function we evaluate \textsc{bigstr} which is similar to
\textsc{str} but is prepended with $1000$ zeroes to artificially inflate compare time. An algorithm
that is more efficient with the number of comparisons it performs should gain an edge there.

The algorithms are evaluated on a variety of input distributions. Shuffled uniformly distributed
values (\textsc{uniform}: $A[i] = i$), shuffled distributions with many duplicates (\textsc{dupsq}:
$A[i] = i \bmod \lfloor \sqrt{n} \rfloor$, \textsc{dup8}: $A[i] = i^8 + n/2 \bmod n$, \textsc{mod8}:
$A[i] = i \bmod 8$, and \textsc{ones}: $A[i] = 1$), partially shuffled uniform distributions
(\textsc{sort50}, \textsc{sort90}, \textsc{sort99} which respectively have the first 50\%,
90\% and 99\% of the elements already in ascending order) and some traditionally notoriously bad
cases for median-of-3 pivot selection (\textsc{organ}: first half of the
input ascending and the second half descending, \textsc{merge}: two equally sized ascending arrays
concatenated). Finally we also have the inputs \textsc{asc}, \textsc{desc} which are inputs that
are already sorted.

The evaluation was performed on an AMD Ryzen Threadripper 2950x clocked at 4.2GHz with 32GB of RAM.
All code was compiled with GCC 8.2.0 with flags \texttt{-march=native -m64 -O2}. To preserve the
integrity of the experiment no two instances were tested simultaneously and no other resource
intensive processes were run at the same time. For all random shuffling a seed was deterministically
chosen for each size and input distribution, so all algorithms received the exact same input for the
same experiment. Each benchmark was re-run until at least 10 seconds had passed and for at least 10
iterations. The former condition reduces timing noise by repeating small instances many times
whereas the latter condition reduces the influence of a particular random shuffle. The mean number
of cycles spent is reported, divided by $n \log_2 n$ to normalize across sizes. In total the
evaluation program spent 9 hours sorting (with more time spent to prepare input distributions).

As the full results are quite large ($12\text{ distributions} \times 3\text{ data types} = 36$
plots), they are included in Appendix A.

\subsection{Results and observations}

First we'd like to note that across all benchmarks \textsc{pdq} and \textsc{bpdq} are never
significantly slower than their non-pattern defeating counterparts \textsc{std} and \textsc{bq}.
In fact, the only regression at all is $\approx 4.5\%$ for \textsc{pdq} v.s. \textsc{std} for
large instances of \textsc{uniform-int} (while being faster for smaller sizes and \textsc{bpdq}
being roughly twice as fast).

There is one exception, with \textsc{bq} beating \textsc{bpdq} for the \textsc{organ} and
\textsc{merge} distributions (and very slightly for \textsc{sort99}) with the \textsc{bigstr} data
type. It is unclear why exactly this happens here. Especially curious is that the former two
distributions were very clearly bad cases for \textsc{bq} with the \textsc{int} data type when we
compare the performance against \textsc{uniform}. But in \textsc{bigstr} the roles have been
reversed, \textsc{bq} is significantly \textit{faster} on \textsc{organ} and \textsc{merge} than it
is on \textsc{uniform}. So it's not that \textsc{bpdq} is slow here, \textsc{bq} is just
mysteriously fast. Regardless, we see that both \textsc{pdq} and \textsc{bpdq} maintain good
performance (similar to \textsc{uniform}) on these cases, effectively defeating the pattern.

With these observations it's safe to say that the heuristics used in pattern-defeating quicksort
come with minimal to no overhead.

Interesting to note is the behavior regarding cache. On this system with
\texttt{sizeof(std::string)} being $32$ we see a drastic change of slope for all quicksort based
algorithms in the \textsc{str} benchmark around $n = 2^{16}$, which is right when the $1.5$MB L1 cache
has filled up. It's odd to see that this shift in slope never happens for \textsc{int}, even when
the input size well exceeds any CPU cache. For \textsc{bigstr} the change of slope occurs slightly
earlier, at around $n = 2^{12}$. Timsort is seemingly barely affected by this, but the clear winner
in this regard is \textsc{is$^4$o}, which appears to be basically cache-oblivious.

We already knew this, but block partitioning isn't always beneficial. However even for types where
the comparison can be done without branches, it still isn't always faster. In cases where the branch
predictor would get it right nearly every time because the data has such a strong pattern (e.g. for
\textsc{pdq} in \textsc{merge-int}), the traditional partitioning can still be significantly faster.

It should come at no surprise that for every input with long ascending or descending runs Timsort is
in the lead. Timsort is based on mergesort, so it can fully exploit any runs in the data. We note
that while pdqsort \textit{defeats} the patterns, meaning it doesn't heavily slow down on
unfavorable patterns, it can't exploit these runs either.

In previous benchmarks however Timsort's constant factor was very high, making it significantly
slower for anything that does not have a pattern to exploit. Looking at the results now, we
congratulate Timothy van Slyke on his performant implementation, which is significantly more
competitive. Especially for bigger or harder to compare types Timsort is now an excellent choice.

Our scheme for handling equal elements is very effective. Especially when the number of equivalence
classes approaches $1$ (such as in \textsc{mod8} and \textsc{ones}) the runtime goes down
drastically, but even in milder cases such as \textsc{dupsq} we see that the pattern-defeating sorts
build a sizable lead over their counterparts when compared to their performance in \textsc{uniform}.

Finally we note that due to the optimistic insertion sort on swapless partitions pdqsort achieves an
incredible speedup for the the ascending and descending input case, rivalling Timsort.

\section{Conclusion and further research}

We conclude that the heuristics and techniques presented in this paper have little overhead, and
effectively handle various input patterns. Pattern-defeating quicksort is often the best choice of
algorithm overall for small to medium input sizes or data type sizes. It and other quicksort
variants suffer from datasets that are too large to fit in cache, where \textsc{is$^4$o} shines. The
latter algorithm however suffers from bad performance on smaller sizes, future research could
perhaps combine the best of these two algorithms.

\pagebreak

\bibliographystyle{splncs04} 
\bibliography{pdqsort}

\appendix
\section{Full experimental evaluation}

All experiments ran are included here, for all three data types. All plots for the same data type
share the same axes, and the plots for each data type are split into two figures, one per page.

\begin{figure}
\includegraphics[width=0.98\textwidth]{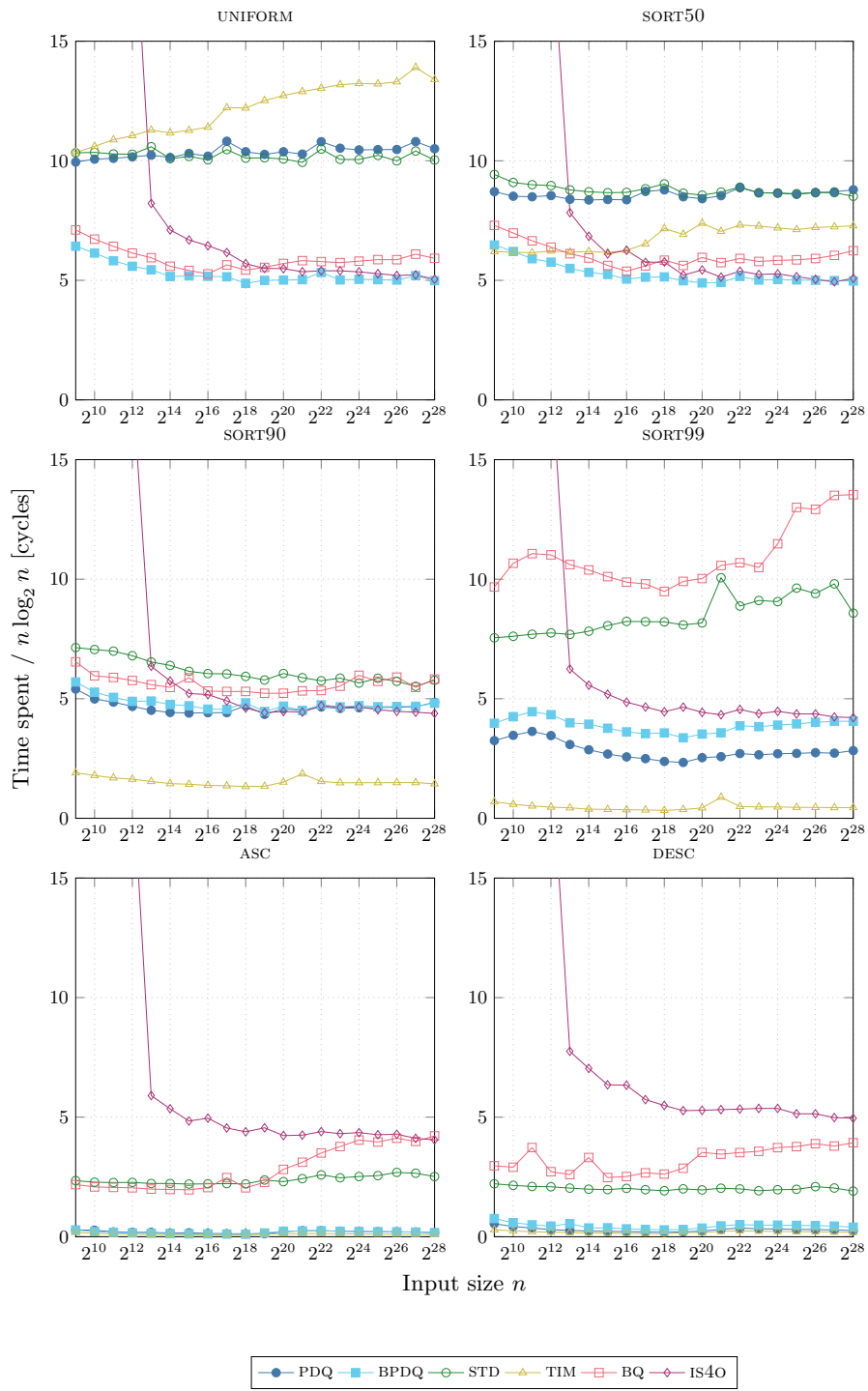}
\caption{\textsc{Int} benchmark.}
\end{figure}

\begin{figure}
\includegraphics[width=0.98\textwidth]{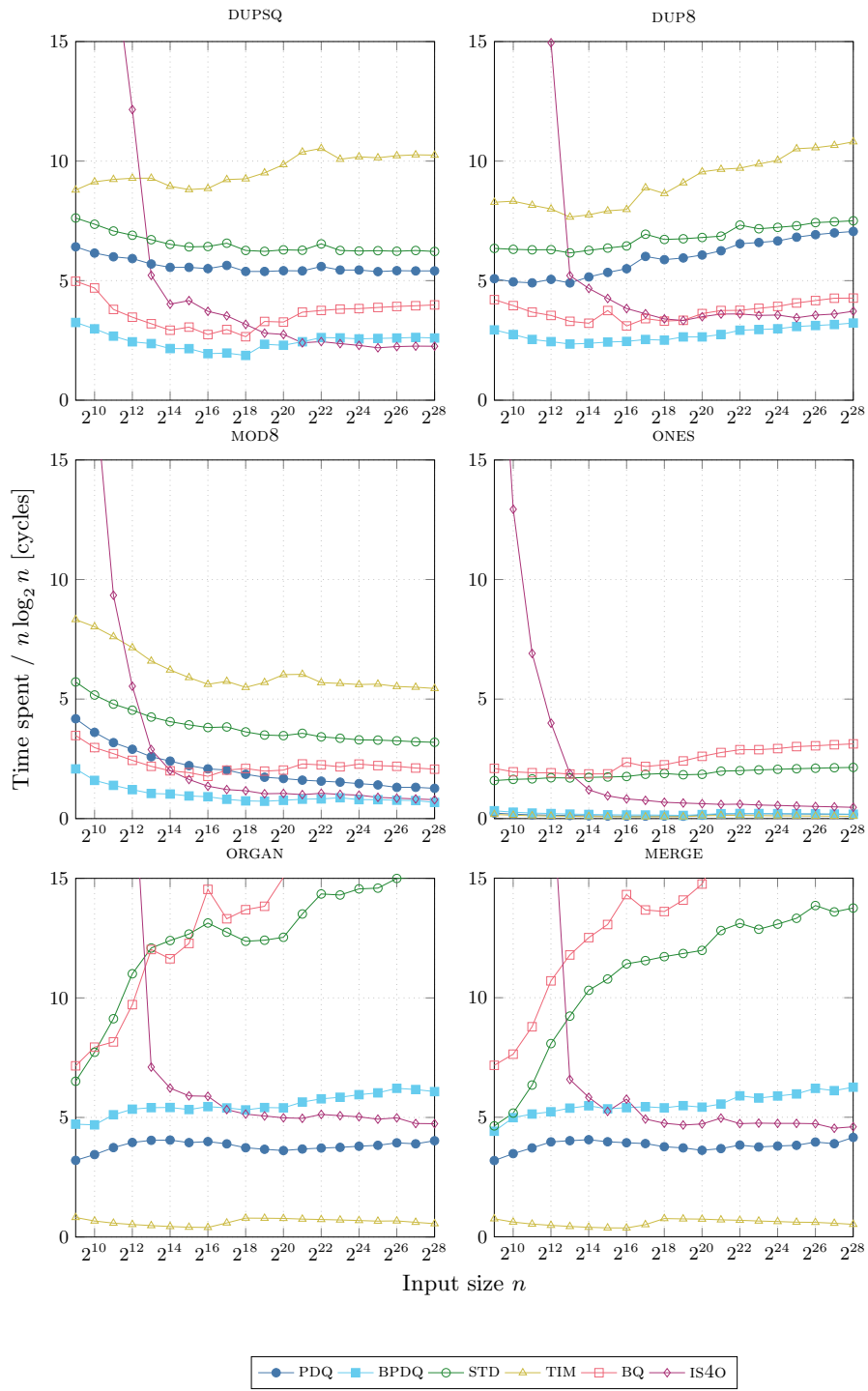}
\caption{\textsc{Int} benchmark, continued.}
\end{figure}

\begin{figure}
\includegraphics[width=0.98\textwidth]{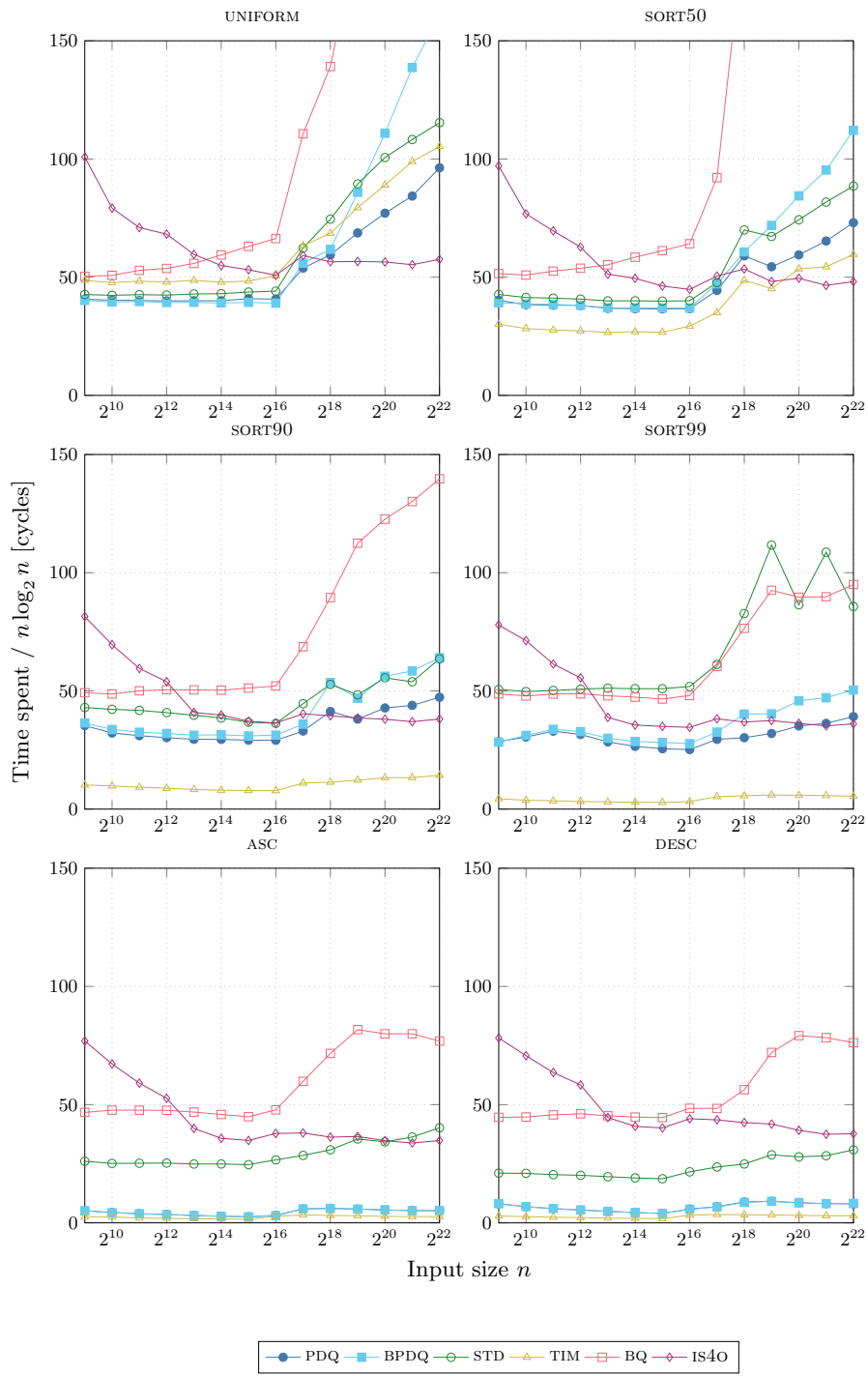}
\caption{\textsc{str} benchmark.}
\end{figure}

\begin{figure}
\includegraphics[width=0.98\textwidth]{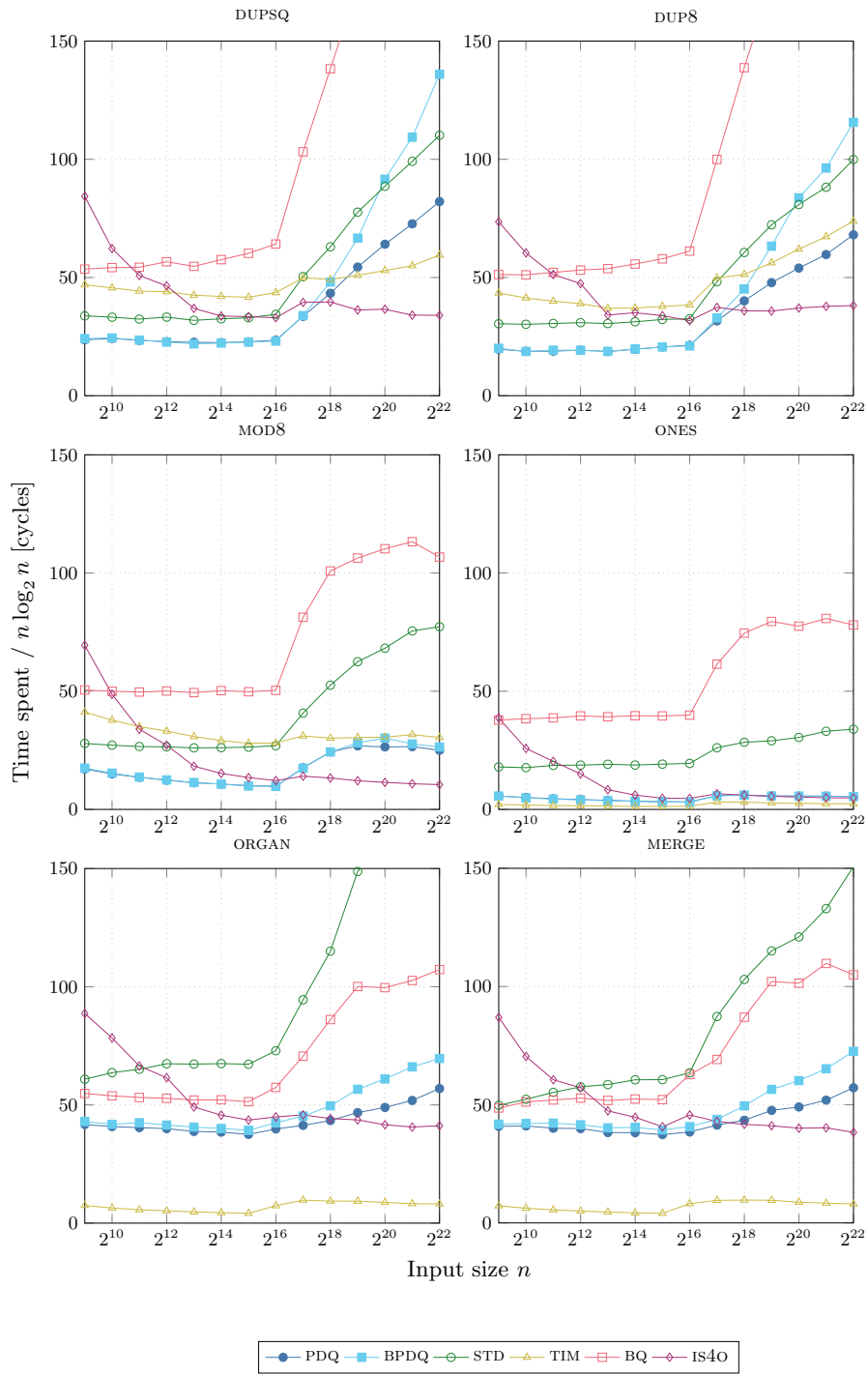}
\caption{\textsc{str} benchmark, continued.}
\end{figure}

\begin{figure}
\includegraphics[width=0.98\textwidth]{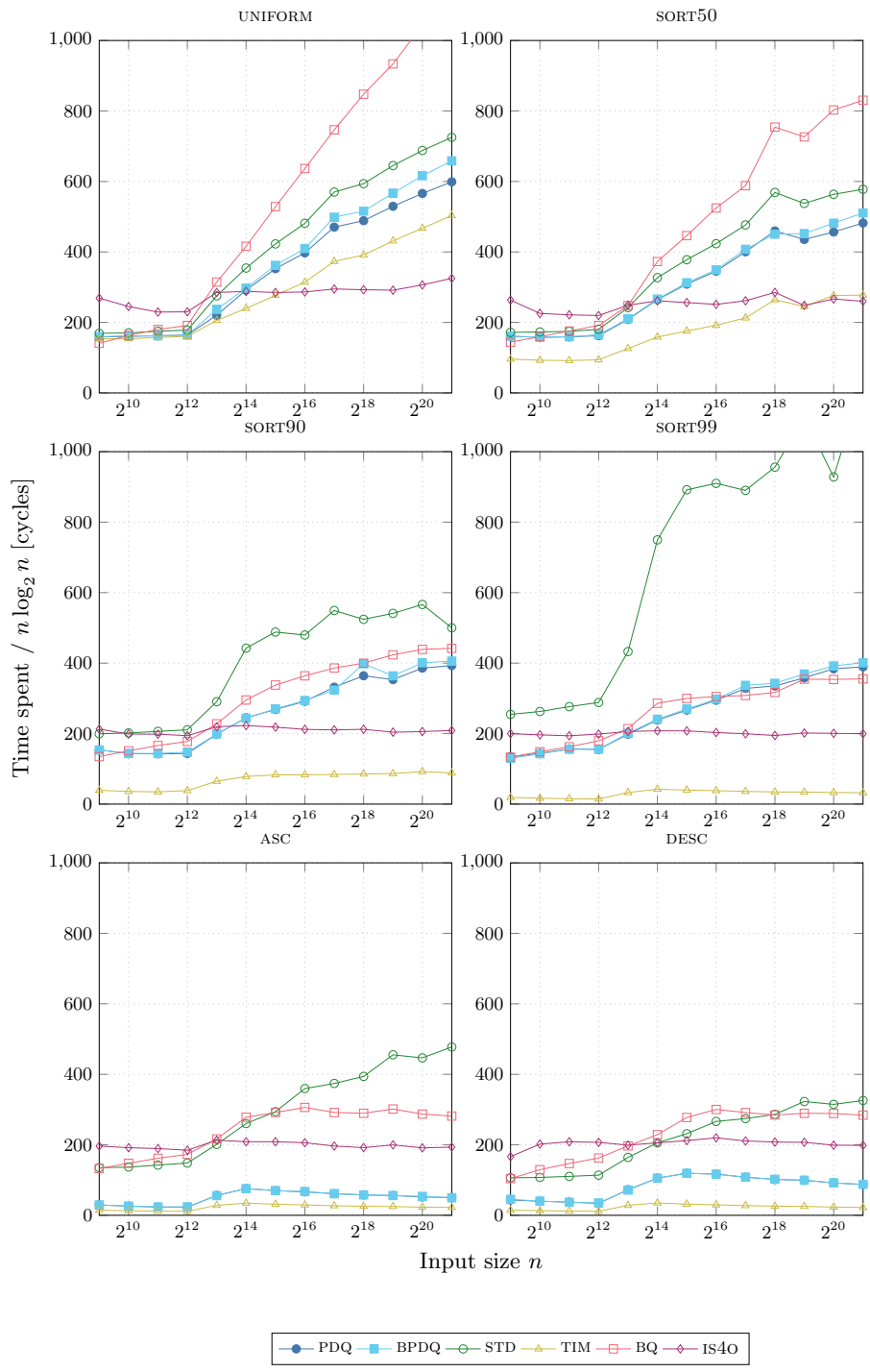}
\caption{\textsc{bigstr} benchmark.}
\end{figure}

\begin{figure}
\includegraphics[width=0.98\textwidth]{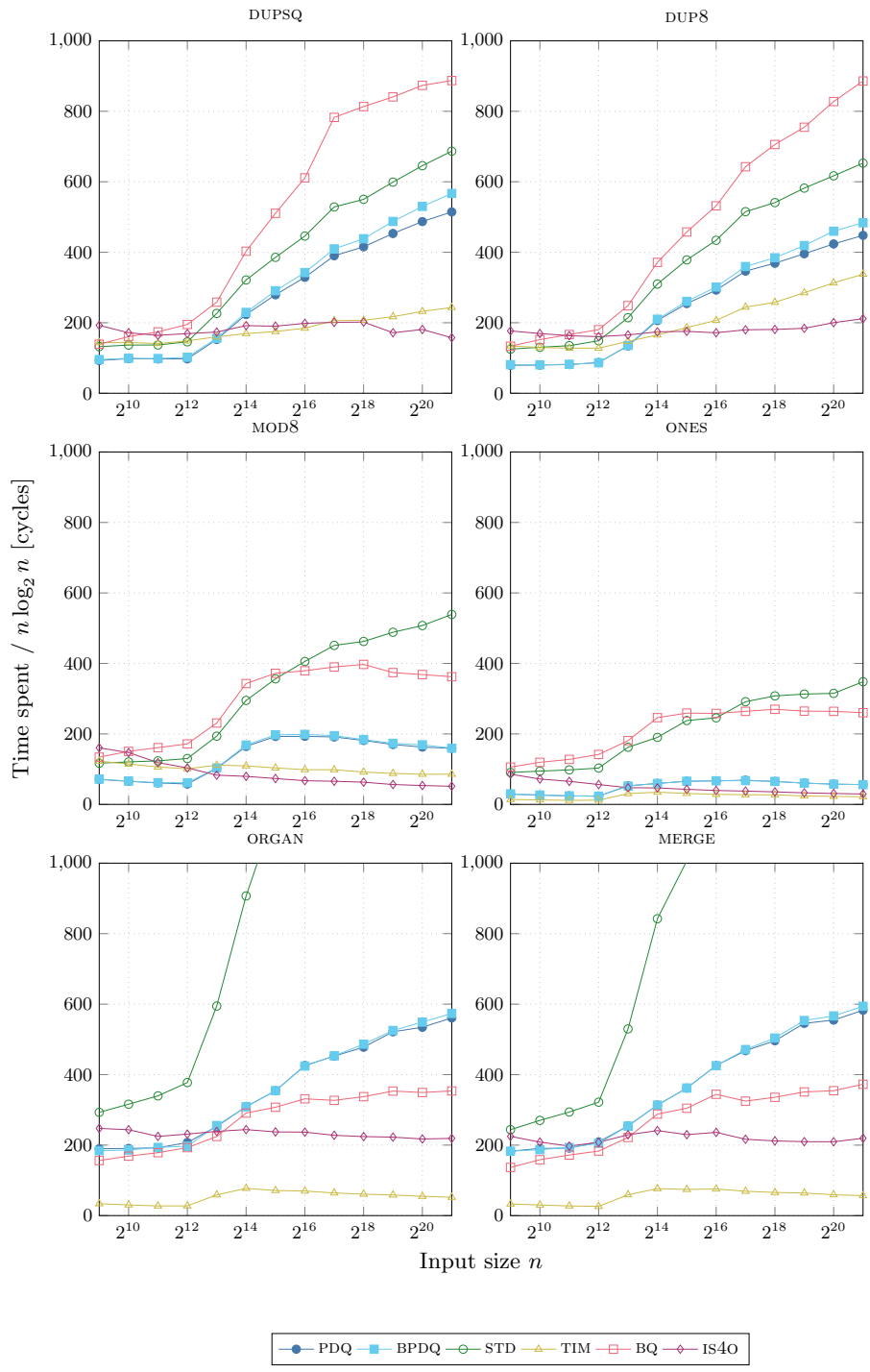}
\caption{\textsc{bigstr} benchmark, continued.}
\end{figure}

\end{document}